\documentclass[12pt]{article}
\usepackage{amsfonts}
\usepackage{amssymb}
\usepackage{amsmath}
\usepackage{amsthm}
\usepackage[colorlinks,linkcolor=blue,citecolor=green]{hyperref}

\newtheorem{theorem}{Theorem}

\newtheorem{remark}[theorem]{Remark}

\newcommand*{\fd}
[2]{\mathchoice{\frac{\delta#1}{\delta#2}}
  {\delta #1/\delta#2}{\delta#1/\delta#2}{\delta#1/\delta#2}}
\begin{document}

\title{\textbf{Classification of bi-Hamiltonian pairs\\
extended by isometries}}
\author{Maxim V. Pavlov$^{1}$ \quad Pierandrea Vergallo$^{2}$ \quad Raffaele
Vitolo$^{2,3}$ \\
$^{1}$Department of Mathematical Physics,\\
Lebedev Physical Institute of Russian Academy of Sciences,\\
Leninskij Prospekt, 53, Moscow, Russia;\\
$^{2}$Department of Mathematics and Physics ``E. De Giorgi'',\\
University of Salento, Lecce, Italy\\
$^{3}$INFN, Section of Lecce, Italy}
\date{{\small \itshape
  Dedicated to Allan Fordy on the occasion of his 70th birthday}}
\maketitle

\begin{abstract}
  The aim of this article is to classify pairs of first-order Hamiltonian
  operators of Dubrovin-Novikov type such that one of them has a non-local part
  defined by an isometry of its leading coefficient. An example of such
  bi-Hamiltonian pair was recently found for the constant astigmatism
  equation. We obtain a classification in the case of 2 dependent variables,
  and a significant new example with 3 dependent variables that is an extension
  of a hydrodynamic type system obtained from a particular solution of the WDVV
  equations.
\end{abstract}

\section{Introduction}

The theory of homogeneous first order differential-geometric Poisson
brackets was established by B.A. Dubrovin and S.P. Novikov in 1983 \cite%
{DN83} in the framework of the Hamiltonian formalism for PDEs. Such Poisson
brackets are defined by local differential operators $A=(A^{ij})$ of the
type 
\begin{equation}
A^{ij}=g^{ij}(\mathbf{u})\partial _{x}+\Gamma _{k}^{ij}(\mathbf{u})u_{x}^{k},
\label{eq:2}
\end{equation}%
where $u^{i}=u^{i}(t,x)$ are field variables, $i=1$, \dots , $n$ depending
on two independent variables $t$, $x$. The operator $A$ yields a Poisson
bracket between densities 
\begin{equation*}
\{F,G\}_{A}=\int \fd{F}{u^i}A^{ij}\fd{G}{u^j}\,dx
\end{equation*}%
if and only if, in the non degenerate case $\det (g^{ij})\neq 0$, $g^{ij}$
is symmetric, $\Gamma _{k}^{ij}=-g^{is}\Gamma _{sk}^{j}$ are Christoffel
symbols of the Levi--Civita connection of $g_{ij}$ (the inverse of $g^{ij}$%
), and the tensor $g^{ij}(\mathbf{u})$ is a flat contravariant metric.

The theory of compatible pairs of such Hamiltonian operators was developed
later in a series of publications (see the review paper \cite%
{mokhov98:_sympl_poiss}). We recall that two Hamiltonian operators $A$, $B$,
are said to be compatible if $A+\lambda B$ is a Hamiltonian operator for
every $\lambda \in \mathbb{R}$. The main application of compatible pairs $%
A_{0}$, $B_{0}$ of homogeneous Hamiltonian operators of the type \eqref{eq:2}
is the integrability of the corresponding quasilinear systems of PDEs of the
form 
\begin{equation*}
u_{t}^{i}=V_{j}^{i}(\mathbf{u})u_{x}^{j},
\end{equation*}%
which is provided by Magri's Theorem \cite{Magri:SMInHEq} when the above
system of PDEs is \emph{bi-Hamiltonian}: 
\begin{equation}
u_{t}^{i}=V_{j}^{i}(\mathbf{u})u_{x}^{j}=A_{0}^{ij}\frac{\delta \mathbf{H}%
_{0}}{\delta u^{j}}=B_{0}^{ij}\frac{\delta \mathbf{\tilde{H}}_{0}}{\delta
u^{j}},  \label{eq:4}
\end{equation}%
where $B_{0}^{ij}=\tilde{g}^{ij}(\mathbf{u})\partial _{x}+\tilde{\Gamma}%
_{k}^{ij}(\mathbf{u})u_{x}^{k}$ and 
\begin{equation*}
\mathbf{H}_{0}=\int h_{0}(\mathbf{u})dx,\text{ \ \ }\mathbf{\tilde{H}}%
_{0}=\int \tilde{h}_{0}(\mathbf{u})dx,
\end{equation*}%
where $h_{0}(\mathbf{u})$ and $\tilde{h}_{0}(\mathbf{u})$ are hydrodynamic
conservation law densities.

In some cases, bi-Hamiltonian hydrodynamic type systems can be considered as
the dispersionless limit of integrable bi-Hamiltonian systems, containing
higher order derivatives:%
\begin{equation*}
u_{t}^{i}=A^{ij}\frac{\delta \mathbf{H}}{\delta u^{j}}=B^{ij}\frac{\delta 
\mathbf{\tilde{H}}}{\delta u^{j}},
\end{equation*}%
where%
\begin{equation*}
\mathbf{H}=\int h(\mathbf{u,u}_{x},\mathbf{u}_{xx},...)dx,\text{ \ \ }%
\mathbf{\tilde{H}}=\int \tilde{h}(\mathbf{u,u}_{x},\mathbf{u}_{xx},...)dx
\end{equation*}%
and $A^{ij},B^{ij}$ are non-homogeneous differential operators of arbitrary
orders. In the dispersionless limit
($\partial _{x}\rightarrow \epsilon \partial _{x}$,
$\partial _{t}\rightarrow \epsilon \partial _{t},$ $\epsilon \rightarrow 0$) we
have $A^{ij}\rightarrow A_{0}^{ij},B^{ij}\rightarrow B_{0}^{ij}$ and
$h(\mathbf{u,u}_{x},\mathbf{u}_{xx},\ldots)\rightarrow h(\mathbf{u}),$ $%
\tilde{h}(\mathbf{u,u}_{x},\mathbf{u}_{xx},\ldots)\rightarrow \tilde{h}(%
\mathbf{u})$.

A large number of examples of the above bi-Hamiltonian systems are known.
They can be generated from solutions of the WDVV equations as Frobenius
manifolds, for example \cite{dubrovin98:_flat_froben}.

In our paper, we consider an alternative class of bi-Hamiltonian systems.
According with an observation made by B.A. Dubrovin and S.P. Novikov, one can
introduce the so-called flat coordinates $a^{k}(\mathbf{u})$ such that the
first local Hamiltonian operator $A_{0}^{ij}$ takes the constant form,
i.e. $A_0^{ij} = \eta ^{ij}\partial _{x}$, where $\eta ^{ij}$ is a symmetric
constant non degenerate matrix. In this case, the flat coordinates $a^{k}$ play
the role of \textquotedblleft Liouville coordinates\textquotedblright\ for the
second local Hamiltonian operator $B_{0}^{ij}$, \emph{i.e.} we have
$B_{0}^{ij}=\tilde{\Gamma}^{ji}\partial _{x}+\partial _{x}\tilde{\Gamma}^{ij}$,
where $\tilde{\Gamma}%
^{ij}$ are a set of functions such that the second metric tensor is given by $%
\tilde{g}^{ij}(\mathbf{a})=\tilde{\Gamma}^{ij}(\mathbf{a}) +\tilde{\Gamma}%
^{ji}(\mathbf{a})$, and $(\tilde{\Gamma}^{ij})_{,k}\equiv \tilde{\Gamma}%
_{k}^{ij} = - \tilde{g}^{is}\tilde{\Gamma}_{sk}^{j}$. Here we denote $(%
\tilde{\Gamma}^{ij})_{,k}\equiv \partial \tilde{\Gamma}^{ij}/\partial
a^{k}$. So, we can rewrite the bi-Hamiltonian hydrodynamic type
system~\eqref{eq:4} as
\begin{equation*}
  a_{t}^{i}=\eta ^{ij}\partial_x\frac{\delta \mathbf{H}}{\delta a^{j}} =
  (\tilde{\Gamma}^{ji}\partial _{x}+\partial _{x}\tilde{\Gamma}^{ij})
  \frac{\delta \mathbf{\tilde{H}}}{\delta a^{j}}.
\end{equation*}

It is well-known that the class of local first-order homogeneous Hamiltonian
operators can be extended to include non-local terms. A widely studied
extension is that of Mokhov--Ferapontov and Ferapontov operators (see \cite%
{F95:_nl_ho} and references therein), which is again homogeneous. Later E.V.
Ferapontov introduced and studied (see \cite{ferapontov92:_non_hamil}) a 
\emph{non-homogeneous, non-local} extension of the form
\begin{equation}
B=\tilde{g}^{ij}\partial _{x}+\tilde{\Gamma}_{k}^{ij}u_{x}^{k}+\epsilon
f^{i}\partial _{x}^{-1}f^{j},  \label{eq:5}
\end{equation}%
where $\epsilon $ is a parameter and $(f^{j})$ is a vector field, $%
f^{j}=f^{j}(\mathbf{u})$, that is an infinitesimal isometry of $\tilde{g}%
^{ij}$. All these Hamiltonian operators are applicable for integrable as
well as for non-integrable systems. In our paper we deal with bi-Hamiltonian
structures, which is a significant part in theory of integrable systems. Our
motivating example is given by the constant astigmatism equation 
\begin{equation}
u_{tt}+\left( \frac{1}{u}\right) _{xx}+2=0,  \label{eq:6}
\end{equation}%
whose bi-Hamiltonian structure, after introducing the variable $u_{t}=v_{x}$
and rewriting the equation as the non-homogeneous quasilinear system 
\begin{equation}
u_{t}=v_{x},\qquad v_{t}=-\left( \frac{1}{u}\right) _{x}-2x,  \label{eq:8}
\end{equation}%
was found in \cite{m.v.13:_lagran_hamil} to be 
\begin{gather}
A=%
\begin{pmatrix}
0 & 1 \\ 
1 & 0%
\end{pmatrix}%
\partial _{x}, \\\label{eeqq:1}
B=%
\begin{pmatrix}
2u & 0 \\ 
0 & \frac{2}{u}%
\end{pmatrix}%
\partial _{x}+%
\begin{pmatrix}
1 & 0 \\ 
0 & -\frac{1}{u^{2}}%
\end{pmatrix}%
u_{x}+%
\begin{pmatrix}
0 & -1 \\ 
1 & 0%
\end{pmatrix}%
v_{x}+%
\begin{pmatrix}
0 & 0 \\ 
0 & 2%
\end{pmatrix}%
\partial _{x}^{-1}.
\end{gather}

In this paper we investigate the following problem: \emph{find all
bi-Hamiltonian pairs $A$, $B$ where} 
\begin{equation}  \label{eq:9}
A= 
\begin{pmatrix}
0 & 1 \\ 
1 & 0%
\end{pmatrix}%
\partial_x, \quad B^{ij}=\Gamma ^{ji}\partial _{x}+\partial _{x}\Gamma
^{ij}+\epsilon f^{i}\partial _{x}^{-1}f^{j},
\end{equation}
where $(f^i)$ is an infinitesimal isometry of $g^{ij}=\Gamma^{ji}+\Gamma ^{ij}$
(here $\epsilon $ is an arbitrary parameter). The classification is made with
respect to the action of the group of local diffeomorphisms of the dependent
variables. We observe that isometries of the leading coefficients have been
used in a different classification problem of pairs of compatible \emph{local}
first-order homogeneous operators $A_0$, $B_0$ in \cite{MF2001}.

We give a complete solution to the above problem in the case of $n=2$ dependent
variables. The solution of the system of conditions in the case $%
n=3$ is more complicated and will be dealt with in the future. However, we
provide here an interesting new example: the ``isometric'' extension of a
hydrodynamic type system obtained as a solution of the WDVV equations.

\bigskip

\textbf{Acknowledgements.} We thank E.V. Ferapontov for useful discussions.
MVP was partially supported by the RFBR grant: 18-01-00411. PV and RV
acknowledge the financial support of Dipartimento di Matematica e Fisica `E.
De Giorgi' of the Universit\`a del Salento, GNFM of the Istituto Nazionale
di Alta Matematica \url{http://www.altamatematica.it} and the
financial support of Istituto Nazionale di Fisica Nucleare through its IS
`Mathematical Methods in Non-Linear Physics'
\url{https://web.infn.it/CSN4/index.php/it/17-esperimenti/165-MMNLP-home}.

\section{Preliminaries}

Let us consider the homogeneous operator of the first order: 
\begin{equation}  \label{op1}
A_0^{ij}=g^{ij}\partial_x + \Gamma^{ij}_ku^k_x.
\end{equation}
It is well-known \cite{DN83} that the conditions necessary and sufficient
for $A_0$ to be Hamiltonian are, in the non degenerate case $%
\det(g^{ij})\neq 0$\footnote{%
From now on all operators are assumed to have a non degenerate leading term.}%
, that $g^{ij}$ is symmetric, its inverse $(g_{ij})$ is a flat
pseudo-Riemannian metric and $\Gamma^i_{jk} = - g_{jp}\Gamma^{ip}_k$ are the
Christoffel symbols of the Levi-Civita connection of $g_{ij}$.

Let us consider two first-order homogeneous Hamiltonian operators $A_{0}$
and $B_{0}=\tilde{g}^{ij}\partial _{x}+\tilde{\Gamma}_{k}^{ij}u_{x}^{k}$.
The pair of metrics $g$ and $\tilde{g}$ is said to be \emph{almost compatible%
} if for every linear combination $g_{\lambda }:=g+\lambda \tilde{g}$ we
have 
\begin{equation}
\Gamma _{\lambda ,k}^{ij}=\Gamma _{k}^{ij}+\lambda \tilde{\Gamma}_{k}^{ij};
\label{gam1}
\end{equation}%
the pair of metrics $g$ and $\tilde{g}$ is said to be \emph{compatible} if
and only if, in addition to~\eqref{gam1}, the Riemann curvature tensor $%
R_{\lambda }$ of the metric $g+\lambda \tilde{g}$ splits as the sum of the
Riemann curvature tensors $R$ of $g$ and $\tilde{R}$ of $\tilde{g}$: 
\begin{equation*}
R_{\lambda ,kl}^{ij}=R_{kl}^{ij}+\lambda \tilde{R}_{kl}^{ij}.
\end{equation*}%
It can be proved that the Hamiltonian operators $A_{0}$, $B_{0}$ are
compatible, \emph{i.e.} $A_{0}+\lambda B_{0}$ is a Hamiltonian operator for
every $\lambda $, if and only if the corresponding metrics are compatible 
\cite{mokhov17:_pencil}. In this case we say $A_{0}$, $B_{0}$ to be a \emph{%
bi-Hamiltonian pair}.

Now, let us consider a non-local non-homogeneous operator of the form%
\begin{equation*}
B^{ij}=g^{ij}\partial _{x}+\Gamma _{k}^{ij}u_{x}^{k}+cu_{x}^{i}\partial
_{x}^{-1}u_{x}^{j}+\epsilon f^{i}\partial _{x}^{-1}f^{j},
\end{equation*}%
where $c$, $\epsilon $ are constants and $(f^{i})$ is a vector field, $%
f^{i}=f^{i}(\mathbf{u})$. In \cite{ferapontov92:_non_hamil} it is shown that 
$B$ defines a Poisson bracket if and only if the following conditions are
satisfied:

\begin{enumerate}
\item $g_{ij}$ is a pseudo-Riemannian metric and $g_{ij}$ is compatible with
the connection with Christoffel symbols $\Gamma^{i}_{jk} = -
g_{jp}\Gamma^{ip}_k$;

\item the connection $\Gamma^i_{jk}$ is symmetric and it has constant
curvature $c$;

\item $(f^i)$ is an infinitesimal isometry of $g_{ij}$, or, equivalently, $%
\nabla^if^j+\nabla^jf^i=0$;

\item the cyclic condition 
\begin{equation*}
f^{j}\nabla ^{i}f^{k}+<\text{cyclic}>=0
\end{equation*}%
is fulfilled.
\end{enumerate}

It is clear that the last condition is trivially satisfied for $2$%
-dimensional spaces.

\begin{remark}
\label{sec:preliminaries} We observe that if we require $c=0$ then the
operator $B$ is of the form 
\begin{equation*}
B=B_{0}+\epsilon f^{i}\partial _{x}^{-1}f^{j},
\end{equation*}%
where $B_{0}$ is a local homogeneous first-order Hamiltonian operator.
\end{remark}

Motivated by the example of the constant astigmatism equation~\eqref{eq:6},
in the two-component case ($n=2$) we consider pairs of Hamiltonian operators 
$A$, $B$, where 
\begin{subequations}
\label{eq:10}
\begin{gather}
A= \eta^{ij}\partial_x = 
\begin{pmatrix}
0 & 1 \\ 
1 & 0%
\end{pmatrix}%
\partial_x, \\
B^{ij}=B_0 + f^{i}\partial_{x}^{-1}f^{j} = g^{ij}\partial_x + \Gamma ^{ij}_k
u^k_x + f^{i}\partial_{x}^{-1}f^{j}.
\end{gather}
\end{subequations}
The compatibility conditions of the above operators are given in the
following theorem.

\begin{theorem}\label{th:1}
The above Hamiltonian operators $A$, $B$ form a bi-Hamiltonian pair if and
only if

\begin{enumerate}
\item $A$ is compatible with $B_0$;

\item $(f^i)$ is an isometry of the leading terms of both operators $A$ and $%
B$.
\end{enumerate}
\end{theorem}

\begin{proof}
  Indeed, $A$ is clearly a Hamiltonian operator. Moreover, the operator
  \begin{equation}
    \label{eq:15}
    \lambda A + B = (\lambda\eta^{ij} + g^{ij})\partial_x + \Gamma ^{ij}_k
    u^k_x +  f^{i}\partial_{x}^{-1}f^{j}
  \end{equation}
  should be Hamiltonian for every $\lambda$.  This is equivalent to the
  requirement that the operator
  \begin{equation}
    \label{eq:16}
    (\lambda\eta^{ij} + g^{ij})\partial_x + \Gamma ^{ij}_ku^k_x
  \end{equation}
  is Hamiltonian for  every $\lambda$, which is the first condition of the
  above statement. The second condition comes from the fact that $(f^i)$ must
  be an isometry of the leading metric coefficient $\lambda\eta^{ij} + g^{ij}$
  for every $\lambda\in\mathbb{R}$.
\end{proof}It is easy to realize that the Hamiltonian operators $A$ and $%
B_{0}$ are compatible if and only if 
\begin{subequations}\label{eq:13}
\begin{gather}\label{eq:11}
  \Gamma _{k}^{ij}=\Gamma _{\lambda ,k}^{ij},
  \\\label{eq:12}
  R_{kl}^{ij}=0=R_{\lambda ,kl}^{ij},
\end{gather}
\end{subequations}
where $g_{\lambda }^{ij}=g^{ij}+\lambda \eta ^{ij}$ and $\Gamma _{k}^{ij}$, $%
R_{kl}^{ij}$ are the Christoffel symbols and the Riemannian curvature tensor
of $g$, respectively. These are the equations that will be used in order to
produce a classification of the pairs $A$, $B$ when $n=2$.

\section{The case \texorpdfstring{$n=2$}{n=2}:
  classification}

\label{sec:case-n=2:-class}

In the case $n=2$ we can provide a classification of the bi-Hamiltonian
pairs $A$, $B$ as in \eqref{eq:9} or \eqref{eq:10}.
We denote the dependent variables by $u$ and $v$.

We stress that working with a pair of the form \eqref{eq:9} means that we used
flat coordinates of the first operator in order to obtain it in the form
$A^{ij}=\eta^{ij}\partial_x$, where $(\eta^{ij})$ is a constant non degenerate
symmetric matrix. Then, by linear transformations of the dependent variables,
we further reduced $(\eta^{ij})$ to the `antidiagonal identity' form in
\eqref{eq:9}.  The only remaining coordinate freedom consists in translations
and scalings, and they will be used to reduce the number of parameters in the
canonical forms.

An immediate observation is that the vector field $f=(f^i)$ must be a linear
combination of the following isometries of $\eta^{ij}$:

\begin{enumerate}
\item $f_1=\partial_u$,

\item $f_2=\partial_v$,

\item $f_3=u\partial_u-v\partial_v$,
\end{enumerate}

so that $f=a_1f_1+a_2f_2+a_3f_3$. If $a_3\neq 0$, then by translating $u$
and $v$ (this will preserve $\eta^{ij}$) we can reduce to the case $%
f=u\partial_u-v\partial_v$. Otherwise, if $a_3= 0$ then by complex scaling $%
u\mapsto cu$, $v\mapsto \frac{v}{c}$ we can transform the isometry $f$ to $%
f=\partial_u+\partial_v$, or $f=\partial_u$. This shows that a complete
classification of compatible pairs $A$ and $B$ (up to transformations
preserving $\eta$) reduces to the three distint cases

\begin{enumerate}
\item $f=\partial_u$,

\item $f=\partial_u+\partial_v$,

\item $f=u\partial_u-v\partial_v$.
\end{enumerate}

Our strategy is the following: we find the most general form of metric
leading coefficient $g$ of $B$ for which the selected vector field $f$ is an
isometry, then we apply Theorem~\ref{th:1}. We will check the compatibility
of the local operators $A$ and $B_0$ by the compatibility of the
corresponding metric leading coefficients.

\subsection{Case \texorpdfstring{$f=\partial_u$}{f=partial u}}

In this case the metric can be written as $g^{ij}=g^{ij}(v)$.

\begin{theorem}
\label{case2} If $f=\partial _{u}$, the metric $g^{ij}$ in the operator $B$
is one of the following: 
\begin{align*}
& g_{1}^{ij}=%
\begin{pmatrix}
\frac{\alpha }{v} & \beta \\ 
\beta & v%
\end{pmatrix}%
\qquad \alpha \neq \beta ^{2}, \\
& g_{2}^{ij}=%
\begin{pmatrix}
g^{11}(v) & g^{12}(v) \\ 
g^{12}(v) & 0%
\end{pmatrix}%
\qquad g^{12}(v)\neq 0, \\
& g_{3}^{ij}=%
\begin{pmatrix}
0 & \beta \\ 
\beta & v%
\end{pmatrix}%
\qquad \beta \neq 0,
\end{align*}
where $g^{11}(v)$ and $g^{12}(v)$ are arbitrary functions.
\end{theorem}

Obviously, a similar statement holds for $f=\partial_v$ by simple change of
variable.

\subsection{Case \texorpdfstring{$f=\partial_u+\partial_v$}{f=partial u + partial v}}

The solutions of the conditions are presented in the following Theorem.

\begin{theorem}
If $f=\partial_u+\partial_v$ then $g^{ij}$ is one of the following

\begin{equation*}
g_{4}^{ij}=%
\begin{pmatrix}
f(-u+v) & -f(-u+v)+\beta  \\ 
-f(-u+v)+\beta  & f(-u+v)%
\end{pmatrix}%
\qquad \beta \neq 0,
\end{equation*}%
where $f=f(-u+v)$ is an arbitrary non-constant function; 
\begin{equation*}
g_{5}^{ij}=%
\begin{pmatrix}
-u+v & \beta  \\ 
\beta  & 0%
\end{pmatrix}%
\qquad \beta \neq 0,
\end{equation*}%
\begin{equation*}
g_{6}^{ij}=%
\begin{pmatrix}
\alpha  & \beta  \\ 
\beta  & \gamma 
\end{pmatrix}%
\qquad \alpha \gamma \neq \beta ^{2}.
\end{equation*}
\end{theorem}

\subsection{Case \texorpdfstring{$f=u\partial_u-v\partial_v$}{f=rotation}}

Here the coefficients of $g^{ij}$ depend, in principle, on both variables $u$%
, $v$. Solving the conditions yields the following results.

\begin{theorem}
\label{case3} If $f=u\partial _{u}-v\partial _{v}$ the metric $g^{ij}$ is
one of the following 
\begin{align*}
& g_{7}^{ij}=%
\begin{pmatrix}
\frac{\alpha uv+\epsilon }{v^{2}} & \beta  \\ 
\beta  & 0%
\end{pmatrix}%
\qquad \beta \neq 0, \\
& g_{8}^{ij}=%
\begin{pmatrix}
\frac{u}{v} & \beta  \\ 
\beta  & \frac{\alpha v}{u}%
\end{pmatrix}%
\qquad \alpha \neq \beta ^{2}, \\
& g_{9}^{ij}=%
\begin{pmatrix}
\frac{\alpha }{F} & \beta  \\ 
\beta  & F%
\end{pmatrix}%
\qquad \alpha \neq \beta ^{2},
\end{align*}%
with%
\begin{equation}
F=\frac{\gamma uv+\epsilon +\sqrt{\left( \gamma ^{2}-4\alpha \right)
u^{2}v^{2}+2\gamma \epsilon uv+\epsilon ^{2}}}{2u^{2}}.  \label{eq:18}
\end{equation}
\end{theorem}

\section{Hierarchies in \texorpdfstring{$2$}{2} components}

\label{sec:hier-const-astigm}

It is well-known that a bi-Hamiltonian pair $A$, $B$ defines a sequence of
Poisson commuting conserved quantities $H_{k}$ by means of Magri's recursion
\cite{Magri:SMInHEq}:
\begin{equation}
A^{ij}\frac{\delta H_{k+1}}{\delta u^{j}}=B^{ij}\frac{\delta H_{k}}{\delta
u^{j}}.
\end{equation}
where $H_{1}$, $H_{2}$ are densities of conservation laws for the
bi-Hamiltonian system of PDEs associated with the pair. The first density is
usually taken as a Casimir of one of the operators, \emph{i.e.} a density that
is in the kernel of one of the operators. The sequence of systems of PDEs
\begin{equation}\label{eq:7}
u^i_{t^k} = A^{ij}\fd{H_k}{u^j}
\end{equation}
is then the \emph{integrable hierarchy}; the proper integrable
system is identified as the first non-trivial system in the hierarchy.

It is possible to give a more explicit construction of the integrable
hierarchies generated by a bi-Hamiltonian pair using recursion operators
defined as $R=B\circ A^{-1}$.  This implies that there is no more need to solve
Magri's recursion. In our case it is possible to give an explicit form of
recursion operators associated with the bi-Hamiltonian pairs that we found in
our classification in the case $n=2$. In the following Section we will show by
means of an example how to construct the integrable system which is the
starting element of the integrable hierarchy.
\begin{enumerate}
\item Case $g_{1}^{ij}$. We have ($\alpha $ is an arbitrary constant) 
\begin{equation*}
B=%
\begin{pmatrix}
\frac{\alpha }{v}\partial _{x}-\frac{\alpha }{2v^{2}}v_{x}+\partial _{x}^{-1}
& \beta \partial _{x}+\frac{1}{2}u_{x} \\ 
\beta \partial _{x}-\frac{1}{2}u_{x} & v\partial _{x}+\frac{1}{2}v_{x}%
\end{pmatrix}%
\end{equation*}%
and the recursion operator is 
\begin{equation*}
R=%
\begin{pmatrix}
\beta \partial _{x}+\frac{1}{2}u_{x} & \frac{\alpha }{v}\partial _{x}-\frac{%
\alpha }{2v^{2}}v_{x}+\partial _{x}^{-1} \\ 
v\partial _{x}+\frac{1}{2}v_{x} & \beta \partial _{x}-\frac{1}{2}u_{x}%
\end{pmatrix}%
\partial _{x}^{-1}
\end{equation*}

\item Case $g_{2}^{ij}$. We have (functions $g^{11}(v)$ and $g^{12}(v)$ are
arbitrary): 
\begin{equation*}
B=%
\begin{pmatrix}
g^{11}(v)\partial _{x}+\frac{1}{2}\frac{d}{dv}g^{11}(v)v_{x}+\partial
_{x}^{-1} & g^{12}(v)\partial _{x}+\frac{d}{dv}g^{12}(v)v_{x} \\ 
g^{12}(v)\partial _{x} & 0%
\end{pmatrix}%
\end{equation*}%
then, the recursion operator is 
\begin{equation*}
R=%
\begin{pmatrix}
g^{12}(v)\partial _{x}+\frac{d}{dv}g^{12}(v)v_{x} & g^{11}(v)\partial _{x}+%
\frac{1}{2}\frac{d}{dv}g^{11}(v)v_{x}+\partial _{x}^{-1} \\ 
0 & g^{12}(v)\partial _{x}%
\end{pmatrix}%
\end{equation*}

\item Case $g_{3}^{ij}$. We have 
\begin{equation*}
B=%
\begin{pmatrix}
\partial _{x}^{-1} & \beta \partial _{x}+\frac{1}{2}u_{x} \\ 
\beta \partial _{x}-\frac{1}{2}u_{x} & v\partial _{x}+\frac{1}{2}v_{x}%
\end{pmatrix}%
\end{equation*}%
and the recursion operator 
\begin{equation*}
R=%
\begin{pmatrix}
\beta \partial _{x}+\frac{1}{2}u_{x} & \partial _{x}^{-1} \\ 
v\partial _{x}+\frac{1}{2}v_{x} & \beta \partial _{x}-\frac{1}{2}u_{x}%
\end{pmatrix}%
\partial _{x}^{-1}
\end{equation*}

\item Case $g_{4}^{ij}$ (the function $f(\gamma )$ is arbitrary). For
simplicity, let us substitute $\gamma :=-u+v$, then the operator is 
\begin{equation*}
B=%
\begin{pmatrix}
f(\gamma )\partial _{x}+\frac{f^{\prime }(\gamma )}{2}(v_{x}-u_{x})+\partial
_{x}^{-1} & (-f(\gamma )+c_{1})\partial _{x}+\frac{f^{\prime }(\gamma )}{2}%
(u_{x}-v_{x})+\partial _{x}^{-1} \\ 
(-f(\gamma )+c_{1})\partial _{x}+\frac{f^{\prime }(\gamma )}{2}%
(u_{x}-v_{x})+\partial _{x}^{-1} & f(\gamma )\partial _{x}+\frac{f(\gamma )}{%
2}(v_{x}-u_{x})+\partial _{x}^{-1}%
\end{pmatrix}%
\end{equation*}%
and the recursion operator is 
\begin{equation*}
R=%
\begin{pmatrix}
(-f(\gamma )+c_{1})\partial _{x}+\frac{f^{\prime }(\gamma )}{2}%
(u_{x}-v_{x})+\partial _{x}^{-1} & f(\gamma )\partial _{x}+\frac{f^{\prime
}(\gamma )}{2}(v_{x}-u_{x})+\partial _{x}^{-1} \\ 
f(\gamma )\partial _{x}+\frac{f(\gamma )}{2}(v_{x}-u_{x})+\partial _{x}^{-1}
& (-f(\gamma )+c_{1})\partial _{x}+\frac{f^{\prime }(\gamma )}{2}%
(u_{x}-v_{x})+\partial _{x}^{-1}%
\end{pmatrix}%
\partial _{x}^{-1}
\end{equation*}

\item Case $g_{5}^{ij}$. The operator is 
\begin{equation*}
B=%
\begin{pmatrix}
(-u+v)\partial _{x}-\frac{1}{2}u_{x}+\frac{1}{2}v_{x}+\partial _{x}^{-1} & 
\beta \partial _{x}+\frac{1}{2}v_{x}+\partial _{x}^{-1} \\ 
\beta \partial _{x}-\frac{1}{2}v_{x}+\partial _{x}^{-1} & \partial _{x}^{-1}%
\end{pmatrix}%
\end{equation*}%
and the recursion operator is 
\begin{equation*}
R=%
\begin{pmatrix}
\beta \partial _{x}+\frac{1}{2}v_{x}+\partial _{x}^{-1} & (-u+v)\partial
_{x}-\frac{1}{2}u_{x}+\frac{1}{2}v_{x}+\partial _{x}^{-1} \\ 
\partial _{x}^{-1} & \beta \partial _{x}-\frac{1}{2}v_{x}+\partial _{x}^{-1}%
\end{pmatrix}%
\partial _{x}^{-1}
\end{equation*}

\item Case $g_{6}^{ij}$. The operator is ($\beta $ and $\gamma $ are
arbitrary constants) 
\begin{equation*}
B=%
\begin{pmatrix}
\alpha \partial _{x}+\partial _{x}^{-1} & \beta \partial _{x}+\partial
_{x}^{-1} \\ 
\beta \partial _{x}+\partial _{x}^{-1} & \gamma \partial _{x}+\partial
_{x}^{-1}%
\end{pmatrix}%
\end{equation*}%
and the recursion operator is 
\begin{equation*}
R=%
\begin{pmatrix}
\beta \partial _{x}+\partial _{x}^{-1} & \alpha \partial _{x}+\partial
_{x}^{-1} \\ 
\gamma \partial _{x}+\partial _{x}^{-1} & \beta \partial _{x}+\partial
_{x}^{-1}%
\end{pmatrix}%
\partial _{x}^{-1}
\end{equation*}

\item Case $g_{7}^{ij}$. We have ($\beta $ and $\gamma $ are arbitrary
constants)%
\begin{equation*}
B=%
\begin{pmatrix}
\frac{\alpha uv+\epsilon }{v^{2}}\partial _{x}+\frac{\alpha }{2v}%
u_{x}+\left( \frac{\alpha u}{2v^{2}}-\frac{\alpha uv+\epsilon }{v^{3}}%
\right) v_{x}+u\partial _{x}^{-1}u & \beta \partial _{x}-\frac{\alpha }{2v}%
v_{x}-u\partial _{x}^{-1}v \\ 
\beta \partial _{x}+\frac{\alpha }{2v}v_{x}-v\partial _{x}^{-1}u & v\partial
_{x}^{-1}v%
\end{pmatrix}%
\end{equation*}%
and we obtain the recursion operator:%
\begin{equation*}
R=%
\begin{pmatrix}
\beta \partial _{x}-\frac{\alpha }{2v}v_{x}-u\partial _{x}^{-1}v & \frac{%
\alpha uv+\epsilon }{v^{2}}\partial _{x}+\frac{\alpha }{2v}u_{x}+\left( 
\frac{\alpha u}{2v^{2}}-\frac{\alpha uv+\epsilon }{v^{3}}\right)
v_{x}+u\partial _{x}^{-1}u \\ 
v\partial _{x}^{-1}v & \beta \partial _{x}+\frac{\alpha }{2v}v_{x}-v\partial
_{x}^{-1}u%
\end{pmatrix}%
\partial _{x}^{-1}.
\end{equation*}

\item Case $g_{8}^{ij}$. Let us consider the following ($c$ is an arbitrary
constant): 
\begin{equation*}
B=%
\begin{pmatrix}
\frac{u}{v}\partial _{x}+\frac{1}{2v}u_{x}-\frac{u}{2v^{2}}v_{x}+u\partial
_{x}^{-1}u & \beta \partial _{x}+\frac{\alpha }{2u}u_{x}-\frac{1}{2v}%
v_{x}-u\partial _{x}^{-1}v \\ 
\beta \partial _{x}-\frac{\alpha }{2u}u_{x}+\frac{1}{2v}v_{x}-v\partial
_{x}^{-1}u & \frac{\alpha v}{u}\partial _{x}-\frac{\alpha v}{2u^{2}}u_{x}+%
\frac{\alpha }{2u}v_{x}+v\partial _{x}^{-1}v%
\end{pmatrix}%
.
\end{equation*}%
The recursion operator is 
\begin{equation*}
R=%
\begin{pmatrix}
\beta \partial _{x}+\frac{\alpha }{2u}u_{x}-\frac{1}{2v}v_{x}-u\partial
_{x}^{-1}v & \frac{u}{v}\partial _{x}+\frac{1}{2v}u_{x}-\frac{u}{2v^{2}}%
v_{x}+u\partial _{x}^{-1}u \\ 
\frac{\alpha v}{u}\partial _{x}-\frac{\alpha v}{2u^{2}}u_{x}+\frac{\alpha }{%
2u}v_{x}+v\partial _{x}^{-1}v & \beta \partial _{x}-\frac{\alpha }{2u}u_{x}+%
\frac{1}{2v}v_{x}-v\partial _{x}^{-1}u%
\end{pmatrix}%
\partial _{x}^{-1}.
\end{equation*}

\item Case $g_{9}^{ij}$. The expressions of the Christoffel symbols make the
Hamiltonian operator and the recursion operator too big to be shown here.
\end{enumerate}

\section{The constant astigmatism equation}

\label{sub2}

In this section we show how to construct an integrable system underlying one of
the bi-Hamiltonian pairs that we found so far. The example contains the
constant astigmatism equation as a particular case; the calculation scheme can
be repeated for any of the bi-Hamiltonian pairs in our classification.

Now, let us consider the metric $g^{ij}_1$ in Theorem~\ref{case2} and change
the variables according with the map $u\mapsto v$, $v\mapsto u$; we obtain 
\begin{equation}
g_{1}^{ij}=%
\begin{pmatrix}
u & \beta \\ 
\beta & \frac{\alpha }{u}%
\end{pmatrix}%
\quad \alpha \neq \beta ^{2},\qquad f^{i}=%
\begin{pmatrix}
0 \\ 
1%
\end{pmatrix}%
\end{equation}%
and the Christoffel symbols of $g_1$ are: 
\begin{gather*}
\Gamma _{1}^{12}=\Gamma _{2}^{22}=\Gamma _{2}^{11}=\Gamma _{1}^{21}=0 \\
\Gamma _{2}^{21}=-\Gamma _{2}^{12}=\Gamma _{1}^{11}=\frac{1}{2}; \\
\Gamma _{1}^{22}=-\frac{\alpha }{2u^{2}};
\end{gather*}%
The operator $B$ is:%
\begin{equation*}
B=%
\begin{pmatrix}
u & \beta \\ 
\beta & \frac{\alpha }{u}%
\end{pmatrix}%
\partial _{x}+%
\begin{pmatrix}
\frac{1}{2} & 0 \\ 
0 & -\frac{\alpha }{2u^{2}}%
\end{pmatrix}%
u_{x}+%
\begin{pmatrix}
0 & -\frac{1}{2} \\ 
\frac{1}{2} & 0%
\end{pmatrix}%
v_{x}+\epsilon 
\begin{pmatrix}
0 & 0 \\ 
0 & \partial _{x}^{-1}%
\end{pmatrix}%
\end{equation*}

Let us apply $B$ to the Casimir $-2v$ of $A$: 
\begin{equation*}
B%
\begin{pmatrix}
0 \\ 
-2%
\end{pmatrix}%
=0+%
\begin{pmatrix}
0 \\ 
\frac{\alpha }{u^{2}}%
\end{pmatrix}%
u_{x}+%
\begin{pmatrix}
1 \\ 
0%
\end{pmatrix}%
v_{x}-\epsilon 
\begin{pmatrix}
0 \\ 
2x%
\end{pmatrix}%
\end{equation*}%
\vspace{5mm}

Now, the constant astigmatism equation 
\begin{equation}  \label{aa1}
u_{tt}+\left(\frac{1}{u}\right)_{xx}+2=0
\end{equation}
can be written in the following form: 
\begin{equation*}
\begin{cases}
u_t=v_x \\ 
v_t=-\left(\frac{1}{u}\right)_x-2x%
\end{cases}%
\end{equation*}
and the associated Hamiltonian operator of Dubrovin-Novikov type is $B$
where $\alpha=1$ and $\beta=0$. By substituting we obtain the following
Christoffel symbols: 
\begin{equation*}
\Gamma^{21}_2=-\Gamma^{12}_2=\Gamma^{11}_1=\frac{1}{2}
\end{equation*}
\begin{equation*}
\Gamma^{22}_1=-\frac{\alpha}{2u^2}=-\frac{1}{2u^2}
\end{equation*}

Therefore, we have 
\begin{equation*}
\begin{cases}
u_{t}=v_{x} \\ 
v_{x}=-\left( -\frac{\alpha}{u^{2}}\right) u_{x}-\epsilon x=-\left( \frac{1}{%
u}\right) _{x}-\epsilon x%
\end{cases}%
\end{equation*}%
In particular, we obtained the equation \eqref{aa1} as the second flow of
the hierarchy generated by the bi-Hamiltonian structure defined by $A$ and $%
B $. The Hamiltonian operator $B$ has the expression 
\begin{equation}\label{eeqq:2}
B=%
\begin{pmatrix}
u & 0 \\ 
0 & \frac{1}{u}%
\end{pmatrix}%
\partial _{x}+%
\begin{pmatrix}
\frac{1}{2} & 0 \\ 
0 & -\frac{1}{2u^{2}}%
\end{pmatrix}%
u_{x}+%
\begin{pmatrix}
0 & -\frac{1}{2} \\ 
\frac{1}{2} & 0%
\end{pmatrix}%
v_{x}+%
\begin{pmatrix}
0 & 0 \\ 
0 & \partial _{x}^{-1}%
\end{pmatrix}%
\end{equation}%
and, according with the computations in Section~\ref{sec:hier-const-astigm}
and the change of variables at the beginning of this Section, we obtain the
following recursion operator: 
\begin{gather*}
R=\left( 
\begin{matrix}
u\partial _{x}+\frac{1}{2}u_{x} & -\frac{1}{2}v_{x} \\ 
+\frac{1}{2}v_{x} & \frac{1}{u}\partial _{x}-\frac{u_{x}}{2u^{2}}+\partial
_{x}^{-1}%
\end{matrix}%
\right) \cdot \left( 
\begin{matrix}
0 & 1 \\ 
1 & 0%
\end{matrix}%
\right) \partial _{x}^{-1}= \\
=%
\begin{pmatrix}
-\frac{1}{2}v_{x} & u\partial _{x}+\frac{1}{2}u_{x} \\ 
\frac{1}{u}\partial _{x}-\frac{u_{x}}{2u^{2}}+\partial _{x}^{-1} & +\frac{1}{%
2}v_{x}%
\end{pmatrix}%
\partial _{x}^{-1}.
\end{gather*}

The reader can also refer to \cite{m.v.13:_lagran_hamil}.  Note that after a
simple change of coordinates the operator $B$ in \eqref{eeqq:2} is exactly the
operator \eqref{eeqq:1} in \cite{m.v.13:_lagran_hamil}.

\section{An example in three components: WDVV equation}
\label{sec:an-example-three}

The equations of associativity, or Witten--Dijkgraaf--Verlinde--Verlinde
equations \cite{D96} contain commuting hydrodynamic type systems
\begin{equation*}
a_{t^{k}}^{i}=\eta ^{im}\left( \frac{\partial ^{2}F}{\partial a^{m}\partial
a^{k}}\right) _{x}.
\end{equation*}%
Indeed, the compatibility conditions $%
(a_{t^{k}}^{i})_{t^{n}}=(a_{t^{n}}^{i})_{t^{k}}$ lead directly to the WDVV system%
\begin{equation*}
\frac{\partial ^{3}F}{\partial a^{i}\partial a^{k}\partial a^{p}}\eta ^{pq}%
\frac{\partial ^{3}F}{\partial a^{q}\partial a^{n}\partial a^{j}}=\frac{%
\partial ^{3}F}{\partial a^{i}\partial a^{n}\partial a^{p}}\eta ^{pq}\frac{%
\partial ^{3}F}{\partial a^{q}\partial a^{k}\partial a^{j}}.
\end{equation*}%
The concept of Frobenius manifold \cite{dubrovin98:_flat_froben} is based on
the existence of a second local Hamiltonian structure for these commuting
hydrodynamic type systems. In the three-component case, if
$\eta ^{ij}=\delta ^{i,4-i}$, then%
\begin{equation*}
F=\frac{1}{2}u^{2}w+\frac{1}{2}uv^{2}+f(v,w),
\end{equation*}%
where the function $f(v,w)$ solves a single third order nonlinear differential
equation%
\begin{equation*}
f_{www}=f_{vvw}^{2}-f_{vww}f_{vvv}.
\end{equation*}%
A simple nontrival solution found by B.A. Dubrovin leads to the ansatz%
\begin{equation*}
f=-\frac{1}{16}v^{4}\gamma (w),
\end{equation*}%
which implies the remarkable Chazy equation%
\begin{equation*}
\gamma ^{\prime \prime \prime }=6\gamma \gamma ^{\prime \prime }-9\gamma
^{\prime ^{2}}.
\end{equation*}

In the semi-simple case, the velocity matrices $\eta ^{pq}\frac{\partial ^{3}F%
}{\partial a^{q}\partial a^{k}\partial a^{j}}$ are non-de\-gen\-er\-ate. So, a
generic solution $\gamma (w)$ of the Chazy equation determines three distinct
characteristic roots of the above velocity matrix. Precisely, according to the
construction by B.A. Dubrovin, we have a pair of commuting hydrodynamic type
systems%
\begin{equation*}
u_{t}=\left( -\frac{1}{4}v^{3}\gamma ^{\prime }(w)\right) _{x},\text{ \ }%
v_{t}=\left( u-\frac{3}{4}v^{2}\gamma (w)\right) _{x},\text{ \ }w_{t}=v_{x},
\end{equation*}%
\begin{equation*}
u_{y}=\left( -\frac{1}{16}v^{4}\gamma ^{\prime \prime }(w)\right) _{x},\text{
\ }v_{y}=\left( -\frac{1}{4}v^{3}\gamma ^{\prime }(w)\right) _{x},\text{ \ }%
w_{y}=u_{x}.
\end{equation*}%
The corresponding velocity matrices
\begin{displaymath}
\eta ^{pq}\frac{\partial ^{3}F}{\partial a^{q}\partial a^{k}\partial a^{2}}
\quad\text{and}\quad
\eta ^{pq}\frac{\partial ^{3}F}{\partial a^{q}\partial a^{k}\partial a^{3}}
\end{displaymath}
are non degenerate, with the exception of the particular case
$\gamma (w)=-2/w$. In the latter case, the hydrodynamic type systems reduce to
the form
\begin{gather}\label{eq:1}
u_{t}=\left( -\frac{1}{2}\frac{v^{3}}{w^{2}}\right) _{x},\text{ \ }%
v_{t}=\left( u+\frac{3}{2}\frac{v^{2}}{w}\right) _{x},\text{ \ }w_{t}=v_{x},
\\\label{eq:3}
u_{y}=\left( \frac{1}{4}\frac{v^{4}}{w^{3}}\right) _{x},\text{ \ }%
v_{y}=\left( -\frac{1}{2}\frac{v^{3}}{w^{2}}\right) _{x},\text{ \ }%
w_{y}=u_{x}.
\end{gather}
Their velocity matrices have a single common characteristic root only; such
degenerate cases are very interesting. For instance, hydrodynamic type
systems with a unique characteristic root were recently investigated \cite%
{kodama16:_confl,except,xue20:_quasil_jordan}. Let us consider
the ffirst of the above three-component hydrodynamic type system \eqref{eq:1}:
\begin{equation*}
\begin{split}
& u_{t}=-\frac{3v^{2}}{2w^{2}}v_{x}+\frac{v^{3}}{w^{3}}w_{x}, \\
& v_{t}=u_{x}+\frac{3v}{w}v_{x}-\frac{3v^{2}}{2w^{2}}w_{x}, \\
& w_{t}=v_{x}.
\end{split}%
\end{equation*}

This system is a first example in the theory of bi-Hamiltonian hydrodynamic
type systems where the first Hamiltonian structure has a non degenerate
metric tensor, while the second Hamiltonian structure has a degenerate
metric tensor. This system possesses an \textquotedblleft
isometry\textquotedblright\ extension, i.e.%
\begin{equation}
\begin{split}
& u_{t}=-\frac{3v^{2}}{2w^{2}}v_{x}+\frac{v^{3}}{w^{3}}w_{x}-x, \\
& v_{t}=u_{x}+\frac{3v}{w}v_{x}-\frac{3v^{2}}{2w^{2}}w_{x}, \\
& w_{t}=v_{x}.
\end{split}
\label{eq:20}
\end{equation}

Eliminating $u$ and introducing the potential function $z$ such that $%
w=z_{x} $ and $v=z_{t}$, we obtain a single, new third order integrable
equation:
\begin{equation*}
z_{ttt}=\left( \frac{3z_{t}^{2}}{2z_{x}}\right) _{xt}-\left( \frac{z_{t}^{3}%
}{2z_{x}^{2}}\right) _{xx}-1.
\end{equation*}%

The system \eqref{eq:20} admits a bi-Hamiltonian pair of the type 
\begin{equation*}
A=%
\begin{pmatrix}
0 & 0 & 1 \\ 
0 & 1 & 0 \\ 
1 & 0 & 0%
\end{pmatrix}%
\partial _{x},\quad B^{ij}=g^{ij}\partial _{x}+\Gamma
_{k}^{ij}u_{x}^{k}+\epsilon f^{i}\partial _{x}^{-1}f^{j},
\end{equation*}%
which is the $n=3$ analogue of the bi-Hamiltonian pair in $2$ components %
\eqref{eq:10}. It can be obtained as follows. The metric of the operator $B$
is 
\begin{equation}
g^{ij}=%
\begin{pmatrix}
\frac{v^{3}}{w^{2}} & \frac{-3v^{2}}{2w} & -v+1 \\ 
\frac{-3v^{2}}{2w} & 2v+1 & w \\ 
-v+1 & w & 0%
\end{pmatrix}%
,  \label{12}
\end{equation}%
and the isometry that defines the nonlocal part of $B$ is $f=\partial _{u}$.
Explicitely, the operator is: 
\begin{align*}
B^{ij}& =%
\begin{pmatrix}
\frac{v^{3}}{w^{2}} & \frac{-3v^{2}}{2w} & -v+1 \\ 
\frac{-3v^{2}}{2w} & 2v+1 & w \\ 
-v+1 & w & 0%
\end{pmatrix}%
\partial _{x}+%
\begin{pmatrix}
0 & 1 & 0 \\ 
-1 & 0 & 0 \\ 
0 & 0 & 0%
\end{pmatrix}%
u_{x}+ \\
& \hphantom{ciaociao}+%
\begin{pmatrix}
\frac{3v^{2}}{2w^{2}} & 0 & 0 \\ 
-\frac{3v}{w} & 1 & 0 \\ 
-1 & 0 & 0%
\end{pmatrix}%
v_{x}+%
\begin{pmatrix}
-\frac{v^{3}}{w^{3}} & 0 & 0 \\ 
\frac{3v^{2}}{2w^{2}} & 0 & 0 \\ 
0 & 1 & 0%
\end{pmatrix}%
w_{x}+%
\begin{pmatrix}
\partial _{x}^{-1} & 0 & 0 \\ 
0 & 0 & 0 \\ 
0 & 0 & 0%
\end{pmatrix}%
\end{align*}

The operators $A$ and $B$ are compatible, hence it is possible to write $B$
in Liouville form: 
\begin{equation*}
B^{ij}=\left( r^{ij}+r^{ji}\right) \partial _{x}+\frac{\partial r^{ij}}{%
\partial u^{k}}u_{x}^{k}+f^{i}\partial _{x}^{-1}f^{j}
\end{equation*}%
where 
\begin{equation*}
r^{ij}=%
\begin{pmatrix}
\frac{v^{3}}{2w^{2}} & u & 1 \\ 
-\frac{3v^{2}}{2w}-u & \frac{1}{2}(2v+1) & 0 \\ 
-v & w & 0%
\end{pmatrix}%
\end{equation*}%
Moreover, it is possible to write $B^{ij}$ as follows: 
\begin{equation*}
B^{ij}=\left( \eta ^{is}\frac{\partial H^{j}}{\partial u^{s}}+\eta ^{js}%
\frac{\partial H^{i}}{\partial u^{s}}\right) \partial _{x}+\eta ^{is}\frac{%
\partial ^{2}H^{j}}{\partial u^{s}\partial u^{k}}u_{x}^{k}+f^{i}\partial
_{x}^{-1}f^{j}
\end{equation*}%
where 
\begin{align*}
& H^{1}(u,v,w)=-uv-\frac{v^{3}}{2w}, \\
& H^{2}(u,v,w)=uw+\frac{v^{2}}{2}+\frac{v}{2}, \\
& H^{3}(u,v,w)=w.
\end{align*}%
This comes from the fact that the local part $B_{0}$ of $B$ is the Lie
derivative of $A$ with respect to a vector field, see \cite%
{mokhov02:_compat_dubrov_novik_hamil_operat} for details.

One can see, that the second metric tensor (\ref{12}) can be presented in
the following form%
\begin{equation}
g^{ij}=\eta ^{ij}+\tilde{g}^{ij},  \label{degen}
\end{equation}%
where the metric $\tilde{g}^{ij}$ is degenerate. Thus, we come to the
observation that our bi-Hamiltonian system has a compatible pair of Hamiltonian
structures, where one of them contains a degenerate metric. This is the first
example of this kind in the theory of bi-Hamiltonian systems.

\begin{theorem}
The non-homogeneous hydrodynamic type system associated with the
bi-Hamiltonian pair $A$, $B$ can be obtained by appling the operator $B$ to
the Casimir $u$ of the operator $A$. In particular: 
\begin{equation*}
\begin{pmatrix}
u \\ 
v \\ 
w%
\end{pmatrix}%
_{t}=A^{ij}\frac{\delta F}{\delta u^{i}}=B^{ij}\frac{\delta L}{\delta u^{i}}
\end{equation*}%
where $L=\int {udx}$ and $F=\displaystyle\int {\left( uv+\frac{v^{3}}{2w}-%
\frac{x^{2}}{2}w\right) dx}$.
\end{theorem}

\section{Conclusion}

In this paper we considered the new problem of the classification of compatible
pairs of Hamiltonian operators $A^{ij}=\eta ^{ij}\partial _{x}$ and
$B^{ij}=\tilde{g}^{ij}\partial _{x}+ \tilde{\Gamma}_{k}^{ij}u_{x}^{k}+\epsilon
f^{i}\partial _{x}^{-1}f^{j}$.

If the parameter $\epsilon$ is zero, then we get back to the classical problem
of the description of compatible Hamiltonian operators
$A_{0}^{ij}=\eta ^{ij}\partial _{x}$ and
$B_{0}^{ij}=\tilde{g} ^{ij}\partial _{x}+\tilde{\Gamma}_{k}^{ij}u_{x}^{k}$,
which was deeply investigated in plenty of papers for past almost forty years
(beginning from the seminal article, written by B.A. Dubrovin and S.P. Novikov
in 1983, see \cite{DN83}). We already know a lot of interesting examples of
such bi-Hamiltonian pairs. However, the complete description is determined by
some integrable systems, which were derived and considered, for instance, in
\cite{dubrovin98:_flat_froben}, \cite{ferapontov01}, \cite{mokhov17:_pencil}.
Even in the $2$-component case, such a system is a hydrodynamic type system in
$4$ dependent variables and $2$ independent variables, and a general solution
is not known.

In comparison with the above situation, the search of compatible pairs of
Hamiltonian operators $A^{ij}=\eta ^{ij}\partial _{x}$ and
$B^{ij}=\tilde{g}^{ij}\partial _{x}+\tilde{\Gamma}_{k}^{ij}u_{x}^{k} +\epsilon
f^{i}\partial_{x}^{-1}f^{j}$ when $\epsilon\neq 0$ is a less complicated task.
We have been able to completely solve this problem in the two component case,
and to give a meaningful example in three components (see
Section~\ref{sec:an-example-three}). We are going to continue this
investigation in forthcoming publications.

Here we would just like to mention some nontrivial byproducts of our
classification. If $\epsilon =0$, we obtain new examples of compatible
Hamiltonian operators $A_{0}^{ij}=\eta ^{ij}\partial _{x}$ and
$B_{0}^{ij}=\tilde{g}^{ij}\partial _{x}+\tilde{\Gamma}_{k}^{ij}u_{x}^{k}$. This
means that in the general case a compatible pair of \emph{local} Hamiltonian
operators of the type of $A_0$ and $B_0$ cannot be extended to the case
$\epsilon \neq 0$. So, our approach allows to construct new distinguished
bi-Hamiltonian structures $A_{0}^{ij}=\eta ^{ij}\partial _{x}$ and
$B_{0}^{ij}=\tilde{g}^{ij}\partial_{x}+\tilde{\Gamma}_{k}^{ij}u_{x}^{k}$.
Moreover, usually, compatible pairs of these bi-Hamiltonian structures were
investigated just if the second metric $\tilde{g}^{ij}$ is
non degenerate. However, in our paper we found a list of new examples where
such a metric is degenerate, see, for instance, (\ref{degen}).


\providecommand{\cprime}{\/{\mathsurround=0pt$'$}} \providecommand*{\SortNoop%
}[1]{}

\end{document}